\documentclass[final,1p,times,authoryear]{elsarticle}

\usepackage{amssymb}
\usepackage{url}
\usepackage{color}
\usepackage{rotating}
\usepackage{rotfloat}
\usepackage{amsmath}
\usepackage{amsthm}
\newtheorem{thm}{Theorem}
\newdefinition{defn}{Definition}
\usepackage{graphicx}
\usepackage{mathtools,bbm}

\journal{Journal of Theoretical Biology}

\begin{document}

\begin{frontmatter}

\title{Conditions for duality between fluxes and concentrations in biochemical networks}
\author[lcsb]{Ronan M.T. Fleming\corref{cor1}}
\ead{ronan.m.t.fleming@gmail.com}

\author[adobe]{Nikos Vlassis}
\ead{nikos.vlassis@gmail.com}

\author[lcsb]{Ines Thiele}
\ead{ines.thiele@gmail.com}

\author[stanford]{Michael A. Saunders}
\ead{saunders@stanford.edu}

\cortext[cor1]{Corresponding author.}
\address[lcsb]{Luxembourg Centre for Systems Biomedicine, University of Luxembourg,
7 avenue des Hauts-Fourneaux, Esch-sur-Alzette, Luxembourg.}

\address[adobe]{Adobe Research, 345 Park Ave, San Jose, CA, USA.}

\address[stanford]{Dept of Management Science and Engineering,
  Stanford University, Stanford, CA, USA.}


\begin{abstract}
Mathematical and computational modelling of biochemical networks is
often done in terms of either the concentrations of molecular species
or the fluxes of biochemical reactions. When is mathematical modelling
from either perspective equivalent to the other? Mathematical duality
translates concepts, theorems or mathematical structures into other
concepts, theorems or structures, in a one-to-one manner. We present
a novel stoichiometric condition that is necessary and sufficient
for duality between unidirectional fluxes and concentrations. Our
numerical experiments, with computational models derived from a range
of genome-scale biochemical networks, suggest that this flux-concentration
duality is a pervasive property of biochemical networks. We also provide
a combinatorial characterisation that is sufficient to ensure flux-concentration
duality. That is, for every two disjoint sets of molecular species,
there is at least one reaction complex that involves species from
only one of the two sets. When unidirectional fluxes and molecular
species concentrations are dual vectors, this implies that the behaviour
of the corresponding biochemical network can be described entirely
in terms of either concentrations or unidirectional fluxes.
\end{abstract}


\end{frontmatter}






\section{Introduction}

Systems biochemistry seeks to understand biological function in terms
of a network of chemical reactions. Systems biology is a broader field,
encompassing systems biochemistry, where understanding is in terms
of a network of interactions, some of which may not be immediately
identifiable with a particular chemical or biochemical reaction. Mathematical
and computational modelling of biochemical reaction network dynamics
is a fundamental component of systems biochemistry. Any genome-scale
model of a biochemical reaction network will give rise to a system
of equations with a high-dimensional state variable, e.g., there are
at least 1000 genes in \emph{Pelagibacter ubique} \citep{giovannoni_genome_2005},
the smallest free-living microorganism currently known. In
order to ensure that mathematical and computational modelling remains
tractable at genome-scale, it is important to focus research effort
on the development of robust algorithms with time complexity that
scales well with the dimension of the state variable.

Given some assumptions as to the dynamics of a biochemical network,
a mathematical model is defined in terms of a system of equations.
Characterising the mathematical properties of such a system of equations
can lead directly or indirectly to insightful biochemical conclusions.
Directly, in the sense that the recognition of the mathematical property
has direct biochemical implications, e.g., the correspondence between
an extreme ray of the steady state (irreversible) flux cone and the
minimal set of reactions that could operate at steady state \citep{schuster_general_2000}.
Or indirectly, in the sense of an algorithm tailored to exploit a
recognised property, which is subsequently implemented to derive biochemical
conclusions from a computational model, e.g., robust flux balance
analysis algorithms \citep{robustFBA} applied to investigate codon
usage in an integrated model of metabolism and macromolecular synthesis
in \emph{Escherichia coli} \citep{Thiele:2009d}. 

Mathematical duality translates concepts, theorems
or mathematical structures into other concepts, theorems or structures
in a one-to-one manner. Sometimes, recognition of mathematical duality
underlying a biochemical network modelling problem enables the dual
problem to be more efficiently solved. An example of this is the
problem of computing \emph{minimal cut sets},
i.e., minimal sets of reactions whose deletion will block the operation
of a specified objective in a steady state model of a biochemical
network \citep{klamt_minimal_2004}. Previously, computation of minimal
cut sets required enumeration of the extreme rays of part of the steady
state (irreversible) flux cone, which is computationally complex in
memory and processing time \citep{haus_computing_2008}. By recognising
that minimal cut sets in a primal network are dual to extreme rays
in a dual network \citep{ballerstein_minimal_2012}, one can compute
select subsets of extreme rays for the dual network that correspond
to minimal cut sets with the certain desired properties in the primal
(i.e., original) biochemical network in question \citep{von_kamp_enumeration_2014}.
This fundamental work has many experimental biological applications,
including metabolic engineering \citep{mahadevan_genome-scale_2015}.

Recognition of mathematical duality in a biochemical
network modelling problem can have many theoretical biological applications,
in advance of experimental biological applications. For example, in
mathematical modelling of biochemical reaction networks, there has
long been an interest in the relationship between models expressed
in terms of molecular species concentrations and models expressed
in terms of reaction fluxes. When concentrations or \emph{net}
fluxes are considered as independent variables, a duality between
the corresponding Jacobian matrices has been demonstrated \citep{jamshidi_flux-concentration_2009}.
In this case, the concentration and net flux Jacobian matrices can
be used to estimate the dynamics of the same system, with respect
to perturbations to concentrations or net fluxes about a given steady
state. The primal (concentration) Jacobian and dual (net flux) Jacobian
matrices are identical, except that one is the transpose of the
other. Matrix transposition is a one-to-one mapping and the aforementioned
duality is between the pair of Jacobians. This does not mean
that the net flux and concentration vectors are dual variables in
the same mathematical sense, and neither are the perturbations to
concentrations or net fluxes. This is because the Jacobian duality
\citep{jamshidi_flux-concentration_2009}, which exists for any stoichiometric
matrix, does not enforce a one-to-one mapping between concentrations
and net fluxes unless the stoichiometric matrix is invertible, which
is never the case for a biochemical network \citep{heinrich_metabolic_1978}.

Herein we ask and answer the question: what conditions are necessary
and sufficient for duality between unidirectional fluxes and molecular
species concentrations? We establish a necessary linear algebraic
condition on reaction stoichiometry in order for duality to hold.
We also
combinatorially characterise this stoichiometric condition in a manner
amenable to interpretation for biochemical networks in general. In
manually curated metabolic network reconstructions, across a wide
range of species and biological processes, we confirm satisfaction
of this stoichiometric condition for the major subset of molecular
species within each reconstruction of a biochemical network. Furthermore,
we demonstrate how linear algebra can be applied to test for satisfaction
of this stoichiometric condition or to identify the molecular
species involved in violation of this condition. We also demonstrate
that violation of flux-concentration duality points to discrepancies
between a reconstruction and the underlying biochemistry, thereby
establishing a new stoichiometric quality control procedure to select
a subset of a biochemical network reconstruction for use in computational
modelling of steady states. 

First, we establish a linear algebraic condition
and a combinatorial condition for duality between
unidirectional fluxes and concentrations. Subsequently, we introduce
a procedure to convert a reconstruction into a computational model
in a quality-controlled manner. We then apply this procedure to a
range of genome-scale metabolic network reconstructions and test for
the linear algebraic condition for flux-concentration duality before
and after conversion into a model. We conclude with a broad discussion,
with examples illustrating how a recognition of flux-concentration
duality could help address questions of biological relevance and improve
our understanding of biological phenomena.

\section{Results }

\subsection{\label{sub:Stoichiometry-=000026-elementary}Stoichiometry and reaction
kinetics}

We consider a biochemical network with $m$ molecular species and
$n$ (net) reactions. Without loss of generality with respect to genome-scale
biochemical networks, we assume $m\leq n$. We assume that each reaction
is \emph{reversible} \citep{lewis1925new} and can be represented
by a \emph{unidirectional reaction} pair. With respect to the forward
direction, in a \emph{forward stoichiometric matrix} $F\in\mathbb{\mathbb{Z}}^{m\times n}$,
let $F_{ij}$ be the \emph{stoichiometry} of molecule $i$ participating
as a substrate or catalyst in \emph{forward unidirectional reaction}
$j$. Likewise, with respect to the reverse direction, in a \emph{reverse
stoichiometric matrix} $R\in\mathbb{\mathbb{Z}}^{m\times n}$, let
$R_{ij}$ be the stoichiometry\emph{ }of molecule $i$ participating
as a substrate or catalyst in \emph{reverse unidirectional reaction}
$j$. 
The set of molecular species that jointly participate
as either substrates or products in a single unidirectional reaction
is referred to as a \emph{reaction complex}. 

One may define the topology of a \emph{hypergraph} of reactions with
a \emph{net stoichiometric matrix} $S:=R-F$. However, a catalyst,
by definition, participates in a reaction with the same stoichiometry
as a substrate or product ($F_{ij}=R_{ij})$, so the corresponding
row of $S$ is all zeros unless that catalyst is synthesised or
consumed elsewhere in the same biochemical network, as is the case
for many biochemical catalysts \citep{thiele2009gcr}. For
example, consider the $i$th molecular species acting as a catalyst
in some reactions. If it is synthesised in the $j$th reaction of
a biochemical network, the stoichiometric coefficient in the forward
stoichiometric matrix will be less than that of the forward stoichiometric
matrix ($F_{ij}<R_{ij}$), so $S_{ij}:=R_{ij}-F_{ij}>0$. This
also encompasses the case of an auto-catalytic reaction.

Before proceeding, some comments on our assumptions
are in order. One may derive $S$ from $F$ and $R$, but the
latter pair of matrices cannot, in general, be derived from $S$ because
$S$ omits the stoichiometry of catalysis. The orientation of the hypergraph,
i.e., the assignment of one direction to be forward (substrates
$\rightharpoonup$ products),
with the other reverse, is typically made so that net flux is forward
(with positive sign) when a reaction is active in its biologically
typical direction in a biochemical network. This
is an arbitrary convention rather than a constraint, and reversing
the orientation of one reaction only exchanges one column of $F$
for the corresponding one in $R$.  Although every chemical reaction
is in principle reversible, in a biochemical setting, due to physiological
limits on the relative concentrations of reactants and substrates,
some reactions are practically irreversible \citep{noor_consistent_2013}.
Our conclusions also extend to systems of irreversible
reactions because the reaction complexes for an irreversible reaction
are the same as those for a reversible reaction.

In the following, the exponential or natural logarithm of a vector
is meant component-wise, with $\exp(\log(0)):=0$. Let $v_{f}\in\mathbb{R}_{>0}^{n}$
and $v_{r}\in\mathbb{R}_{>0}^{n}$ denote forward and reverse unidirectional
reaction rate vectors. We assume that the rate of a unidirectional
reaction is proportional to the product of the concentrations of each
substrate or catalyst, each to the power of their respective stoichiometry
in that unidirectional reaction \citep{wilhelmy1850Kinetics}, with
linear proportionality given by strictly positive\emph{ rate coefficients}
$k_{f},k_{r}\in\mathbb{R}_{>0}^{n}$. Therefore we have 
\begin{equation}
\begin{array}{c}
v_{f}(c):= \exp(\ln(k_{f})+F^{T}\ln(c)),\\
v_{r}(c):= \exp(\ln(k_{r})+R^{T}\ln(c)),
\end{array}\label{eq:elementaryKinetics}
\end{equation}
where $c\in\mathbb{R}_{\geq0}^{m}$ are molecular species concentrations.
Strictly, it is not proper to take the logarithm of a unit that has
physical dimensions, so $c$ should be termed a vector of mole fractions
rather than concentrations \citep[Eq.~19.93]{physChemBRR}, but
safe in the knowledge that we have taken this liberty, we continue
in terms of concentrations. 

If the $j$th columns of $F$ and $R$ represent the stoichiometry
of an \emph{elementary reaction}, then the respective $j$th unidirectional
reaction rate is given by an elementary kinetic rate law in (\ref{eq:elementaryKinetics}).
In biochemical modelling, often it is \emph{composite reaction} stoichiometry
that is represented, in which case the unidirectional reaction rates
are given by pseudo-elementary kinetic rate laws. We shall revisit
this point in discussion, but for now it suffices to mention that,
in principle, all composite reactions can be decomposed into a set
of elementary reactions following elementary reaction kinetics \citep{cook2007eka},
even allosteric reactions \citep{bray_conformational_2004}.
With respect to the forward direction of an elementary reaction, the
term \emph{reaction complex} implies a corresponding physical association
between substrate molecular species. For the sake of simplicity, we
also use the term reaction complex for composite reactions, as if
there were a corresponding simultaneous physical association of all
substrates, which is generally not the case because composite reactions
occur as a set of elementary reaction steps.

With respect to time, the deterministic rate of change of concentration
is
\begin{eqnarray}
\frac{dc}{dt} & := & (R-F)(v_{f}(c)-v_{r}(c)),\label{eq:established_dcdt}\\
 & = & ([R,\,F]-[F,\,R])\left[\begin{array}{c}
v_{f}(c)\\
v_{r}(c)
\end{array}\right],
\end{eqnarray}
where $v_{f}(c)-v_{r}(c)$ gives a vector of net reaction rates,
$\left[\cdot,\,\cdot\right]$ denotes the horizontal concatenation operator,
and $:=$ denotes ``is defined to be equal to''.  Time-invariant fluxes
or concentrations satisfy (\ref{eq:established_dcdt}) with ${dc}/{dt}:=0$.
Define $k:=\left[\begin{array}{c}
k_{f}\\
k_{r}
\end{array}\right]\in\mathbb{R}_{>0}^{2n}$ to be given constants, then
consider the \emph{flux function} 
\begin{equation}
v(c):=\exp(\ln(k)+[F,\,R]^{T}\ln(c))=\left[\begin{array}{c}
v_{f}(c)\\
v_{r}(c)
\end{array}\right]\label{eq:fluxFunction}
\end{equation}
with a concentration vector $c$ the only argument. Apart from (a)
our deliberate distinction between unidirectional and net stoichiometry,
(b) our deliberate use of matrix-vector notation, and (c) our deliberate
use of component-wise exponential and logarithm, the expression for
unidirectional rate in (\ref{eq:fluxFunction}) is a standard representation
of deterministic elementary reaction kinetics.

\subsection{\label{sub:Ample-stoichiometry}Linear algebraic characterisation
of flux-concentration duality}

Herein, duality is defined as a one-to-one relationship between two
variable vectors. We now establish a linear algebraic
condition for duality between unidirectional flux and concentration
vectors.  This linear algebraic condition is a well
known result in mathematics, but to our knowledge its application
to establish duality between unidirectional flux and molecular species
concentration is novel.

\begin{thm}
\label{thm:duality}Assume we are given constants $k\in\mathbb{R}_{>0}^{2n}$
and $F,R\in\mathbb{\mathbb{Z}}_{\geq0}^{m\times n}$. Suppose a unidirectional
reaction flux vector $v\in\mathbb{R}_{>0}^{2n}$ and a molecular species
concentration vector $c\in\mathbb{R}_{>0}^{m}$ satisfy 
\begin{equation}
v=\exp(\ln(k)+[F,\,R]^{T}\ln(c)).\label{eq:fluxFunctionInstance}
\end{equation}
Then $\textrm{rank}([F,\,R])=m$ is a necessary and sufficient
condition for duality between fluxes and concentrations.
\end{thm}

\begin{proof}
That $v$ is uniquely defined given $c$ is trivial. Taking the logarithm
of both sides of (\ref{eq:fluxFunctionInstance}), we have $\ln(v)-\ln(k)=[F,\,R]^{T}\ln(c)$.
Then, if and only if $\textrm{rank}([F,\,R])=m$ is $\ln(c)$,
and therefore $c$, uniquely defined given $v$.
\end{proof}

Theorem \ref{thm:duality} establishes that the flux function (\ref{eq:fluxFunction})
is an injective function. It is not bijective because one can always
find a $v$ such that $\ln(v)-\ln(k)$ is not in the range of $[F,\,R]^{T}$.
Note that the exponential function is bijective,
but if one wished to consider other flux functions, it would be sufficient
to replace the exponential function with another injective function
and Theorem \ref{thm:duality} would still hold.

We now proceed to interpret this stoichiometric condition for duality
in biochemical terms. Consider the following triplet of isomerisation
reactions involving three molecular species:
\[
A\rightleftharpoons B,\quad B\rightleftharpoons C,\quad C\rightleftharpoons A.
\]
The forward, reverse and net stoichiometric matrices are

\begin{equation}
F=\left[\begin{array}{ccc}
1 & 0 & 0\\
0 & 1 & 0\\
0 & 0 & 1
\end{array}\right],\quad R=\left[\begin{array}{ccc}
0 & 0 & 1\\
1 & 0 & 0\\
0 & 1 & 0
\end{array}\right],\quad(R-F)=\left[\begin{array}{rrr}
\!\!-1 & 0 & 1\\
1 & -1 & 0\\
0 & 1 & -1
\end{array}\right],\label{eq:ampleStoichiometryExample}
\end{equation}
where flux and concentration vectors are dual vectors because
$\textrm{rank}([F,\,R])=3=m$.
 Consider the following quartet of reactions involving four representatives
of supposedly distinct molecular species:
\[
A\rightleftharpoons B+C,\quad A\rightleftharpoons D,\quad B+C\rightleftharpoons D,\quad A+D\rightleftharpoons2B+2C.
\]
The forward, reverse and net stoichiometric matrices are 
\begin{equation}
F=\left[\begin{array}{cccc}
1 & 1 & 0 & 1\\
0 & 0 & 1 & 0\\
0 & 0 & 1 & 0\\
0 & 0 & 0 & 1
\end{array}\right],\quad R=\left[\begin{array}{cccc}
0 & 0 & 0 & 0\\
1 & 0 & 0 & 2\\
1 & 0 & 0 & 2\\
0 & 1 & 1 & 0
\end{array}\right],\quad(R-F)=\left[\begin{array}{rrrr}
\!\!-1 & -1 & 0 & -1\\
1 & 0 & -1 & 2\\
1 & 0 & -1 & 2\\
0 & 1 & 1 & -1
\end{array}\right],\label{eq:degenerateStoichiometryExample}
\end{equation}
where flux and concentration vectors are \emph{not} dual vectors because
$\textrm{rank}([F,\,R])=3<4=m$. Observe that the second and third
rows of $F$ and $R$ are positive multiples of one another. This
corresponds to a pair of supposedly distinct molecules, $B$ and $C$,
that are always either produced or consumed together with fixed relative
stoichiometry. This is an ambiguous model of reaction stoichiometry
because either (i) $B$ and $C$ are actually the same molecular species
and therefore the extra row is superfluous, or (ii) $B$ and $C$
are different molecular species but the model is missing
some reaction that would demonstrate they are synthesised or consumed
in distinct reactions.

\subsection{\label{sub:Combinatorial-characterisation-o}Combinatorial characterisation
of flux-concentration duality}

The aforementioned linear algebraic condition for
duality between unidirectional flux and concentration vectors is hard
to interpret in terms of reaction complex stoichiometry. Therefore
we sought a characterisation that would be easier to interpret in
a (bio)chemically interpretable manner.  Here we derive a combinatorial
characterisation of the condition $\textrm{rank}([F,\,R])=m$,
which holds independently of the actual values of the stoichiometric
coefficients. Our analysis draws from the theory of L-matrices and
zero/sign patterns \citep{hershkowitz1993ranks,brualdi2009matrices}.
First we introduce some definitions and notation.

\begin{defn}
\label{dfn:support}(Support of a set of vectors) Let $\mathcal{C}$
be a collection of $d$-dimensional row vectors. The support of $\mathcal{C}$
is defined to be the subset of $\mathcal{I}:=\{1,\ldots,d\}$ such
that, for each $i$ in the given subset of $\mathcal{I}$, there exists
at least one vector in $\mathcal{C}$ whose $i$th component is
nonzero. 
\end{defn}

For example, if $\mathcal{C}$ is formed by the last two
rows of the matrix 
\[
\left[\begin{array}{cccccc}
1 & 0 & 1 & 1 & 0 & 0\\
0 & 1 & 0 & 0 & 1 & 0\\
0 & 1 & 1 & 0 & 0 & 1
\end{array}\right],
\]
the support of $\mathcal{C}$ is $\{2,3,5,6\}$. If $\mathcal{C}$
is formed by the first and third columns of the matrix, its support
is $\{1,3\}$.

\begin{defn}
\label{dfn:combinatorialindependence}(Combinatorial independence)
A collection $\mathcal{C}$ of row vectors (of equal dimension) is
said to be combinatorially independent if $\mathcal{C}$ does not
contain the zero vector and \emph{every} two nonempty disjoint subsets
of $\mathcal{C}$ have different supports. 
\end{defn}

In the above example, the rows of the matrix are combinatorially independent.
However, the columns of this matrix are not combinatorially independent
because the support of columns $\{1,2\}$ is $\{1,2,3\}$, which is
the same as the support of columns $\{3,5\}$.

\begin{defn}
\label{dfn:zeropattern}(Zero pattern) The zero pattern of a real
matrix $A$ is the $(0,1)$-matrix obtained by replacing each nonzero
entry of $A$ by $1$.
\end{defn}

\begin{thm}
\label{thm:richman78}(Combinatorial independence and rank\textbf{
}\citep[Lemma (5.2)]{richman1978singular}) Let $P$ be an $m\times d$
zero pattern. Every non-negative matrix with zero pattern $P$ has
rank $m$ if and only if the rows of $P$ are combinatorially independent. 
\end{thm}

Conversely, it follows that if any two disjoint subsets of rows of
$P$ have the same support, then $P$ is row rank-deficient. For example,
the matrix 
\[
\left[\begin{array}{cccccc}
1 & 1 & 0 & 0 & 0 & 0\\
0 & 0 & 1 & 1 & 0 & 0\\
0 & 0 & 0 & 0 & 1 & 1\\
1 & 0 & 0 & 0 & 1 & 0\\
0 & 1 & 1 & 0 & 0 & 0\\
0 & 0 & 0 & 1 & 0 & 1
\end{array}\right]
\]
is row rank-deficient because rows $\{1,2,3\}$ and rows $\{4,5,6\}$
have the same support $\{1,2,3,4,5,6\}$. Theorem \ref{thm:richman78}
permits us to state the following.

\begin{thm}
\label{crl:corollary} (Combinatorial independence and duality) Consider
a family of biochemical networks that share the same zero pattern
as $[F,\,R]$. Assume that each molecular species participates in at
least one reaction in each network in the family. Then, for each network
in the family, the following are equivalent:
\begin{enumerate}
\item The matrix $[F,\,R]$ has full row rank. 
\item For every two disjoint sets of molecular species, there is at least
one reaction complex that involves species from only one of the two
sets.
\item Unidirectional flux and concentration are dual variables.
\end{enumerate}
\end{thm}

Equivalently, for a given biochemical network with matrix $[F,\,R]$,
if condition 2 in Theorem \ref{crl:corollary} is true, then one can
exchange \emph{any} positive stoichiometric coefficient of the network
with \emph{any} positive value and $[F,\,R]$ will still have full
row rank. The above result provides a combinatorial characterisation
of the condition for flux-concentration duality, which holds independent
of the values of the stoichiometric coefficients. This is analogous
to results involving L-matrices for problems such as the structural
controllability of systems \citep{brualdi2009matrices}.

\subsubsection{Testing for combinatorial independence}

According to Theorem \ref{thm:richman78}, to test if an $m\times d$
zero pattern has rank $m$, we can equivalently test whether its $m$
rows are combinatorially independent. Can this test be performed
efficiently?  In general (unless P=NP) the answer is no  
~\citep{klee1984signsolvability}, as
the problem of testing if a \emph{sign pattern} 
(elements $\{0,1,-1\}$) has full row rank is NP-complete.  Their
proof relies on a reduction from the 3-SAT problem, which is known to
be NP-complete \citep{garey1979computers}. In their proof they
construct a \emph{non-negative} sign pattern (which is a zero
pattern), and therefore their result applies to our case too. Hence we
have the following.

\begin{thm}
\label{thm:test}\textbf{ }(Testing combinatorial independence\textbf{
\citep{klee1984signsolvability}}) Let $P$ be a zero pattern. Testing
if the rows of $P$ are combinatorially independent is NP-complete. 
\end{thm}

However, as we prove next, when the zero pattern is constrained to
have at most two non-negative entries per column, the testing for
combinatorial independence can be done in polynomial time. To our
knowledge, this result is new.

\begin{thm}
\label{thm:testconstrained}(Testing combinatorial independence in
constrained zero patterns) Let $P$ be a zero pattern with at most
two 1s per column. Testing if the rows of $P$ are combinatorially
independent can be done in polynomial time.
\end{thm}

\begin{proof}
Without loss of generality we can assume that each column of $P$
has exactly two nonzero entries. We view the matrix $P$ as the incidence
matrix of an undirected graph, where each row of $P$ is a vertex
and each column is an edge. Combinatorial dependence of the rows of
$P$ would imply the existence of two disjoint sets of rows with the
same support, which would imply the existence of a connected component
of the graph that is bipartite (2-colorable). Finding all connected
components of a graph and bipartiteness testing are classical graph
problems that can be solved in polynomial time \citep{Cormen2009}.
\end{proof}

Since most reconstructed biochemical networks are in terms of composite
reactions, the corresponding $[F,\,R]$ may have more than two nonzero
entries per column and the nonzero stoichiometric coefficients may
differ from $1$. However, every composite reaction is a composition
of a set of elementary reactions \citep{cook2007eka}, each with at
most three reactants per reaction, so the resulting \emph{bilinear}
$[F,\,R]$ will have at most two nonzero entries per column. It is
possible to algorithmically convert any composite reaction into a
set of elementary reactions, with at most two nonzero entries per
column, by creating faux molecular species representing a reaction
intermediate, e.g., the composite reaction $A+B\rightleftharpoons C+D$
may be decomposed into $A+B\rightleftharpoons E$ and $E\rightleftharpoons C+D$. Reaction
intermediates are typically not identical for two enzyme-catalysed
composite reactions, suggesting that flux-concentration duality is
a pervasive property of biochemical networks in general.

\subsection{\label{sub:exist_Flux-concentration-duality}Flux-concentration duality
in existing genome-scale biochemical networks}

Section \ref{sub:Combinatorial-characterisation-o} provided a biochemically
interpretable condition, in terms of molecular species involvement
in reaction complex stoichiometry, that implies flux-concentration
duality for an arbitrary network. We now show that flux-concentration
duality is a pervasive property of quality-controlled models derived
from genome-scale biochemical network reconstructions. Testing for
combinatorial independence is computationally complex, so instead
we rely on linear algebra to test the rank of $[F,\,R]$.
As detailed below, we converted 29 genome-scale metabolic network
reconstructions into computational models, then compared the number
of molecular species with the rank of $[F,\,R]$ before
and after conversion. These metabolic reconstructions were all manually
curated and represent a wide range of different species (see Supplementary
Table 1). 

It is important to distinguish a network reconstruction from a computational
model of a biochemical network. The former may contain incomplete
or inconsistent knowledge of biochemistry, while the latter must
satisfy certain modelling assumptions, represented by mathematical
conditions, in order to ensure that the model is a faithful representation
of the underlying biochemistry. This modelling principle is already
well established in the digital circuit modelling community, and some
of the associated \emph{model checking} algorithms have been applied
to biochemical networks \citep{carrillo_overview_2012}, especially
by the community that use Petri-nets to model biochemical networks,
e.g., \citep{soliman_invariants_2012}. The application of modelling
assumptions is a key step in the conversion of a reconstruction into
a computational model. We now introduce these assumptions, their mathematical
representation, and their relationship to the rank of $[F,\,R]$.
For the sake of simplicity, the toy examples given to illustrate key
concepts only involve reactions with two or less reactants, but the
theory presented also applies to systems of composite reactions involving
three or more reactants.

\subsubsection{Stoichiometric consistency}

All biochemical reactions conserve mass; therefore it is essential
in a model that each reaction, which is supposed to represent a biochemical
reaction, does actually conserve mass. Although it is not essential
to do so \citep{fleming2011exist}, reactions that do not conserve
mass are often added to a network reconstruction \citep{thieleTests}
in order to represent the flow of mass into and out of a system, e.g.,
during flux balance analysis \citep{Pal06}. Every reaction that does
not conserve mass, but is added to a model in order represent the
exchange of mass across the boundary of a biochemical system, is henceforth
referred to as an \emph{exchange reaction,} e.g., $D\rightleftharpoons\emptyset$,
where $\emptyset$ represents null. When checking for reactions that
do not conserve mass, we must first omit exchange reactions. 

Besides exchange reactions, a reconstruction may contain reactions
with incompletely specified stoichiometry or molecules with incompletely
specified chemical formulae, because of (for instance) limitations
in the available literature evidence. While stoichiometrically inconsistent
biochemical reactions may appear in a reconstruction, they should
be omitted from a computational model derived from that reconstruction,
especially if the model is to be used to predict flow of mass, else
erroneous predictions could result. One approach is to require that
chemical formulae be collected for each molecule during the reconstruction
process \citep{Thorleifsson2011}, then omit non-exchange reactions
that are elementally imbalanced \citep{cobraV2}. A complementary
approach is to detect reactions that are specified in a \emph{stoichiometrically
inconsistent} manner \citep{gevorgyan2008detection}. For instance,
the reactions $A+B\rightleftharpoons C$ and $C\rightleftharpoons A$
are stoichiometrically inconsistent because it is impossible to assign
a positive molecular mass to all species whilst ensuring that each
reaction conserves mass. 

A set of stoichiometrically consistent reactions is mathematically
defined by the existence of at least one $\ell\in\mathbb{R}_{>0}^{m}$
such that $R^{T}\ell=F^{T}\ell$, equivalently $S^{T}\ell=(R-F)^{T}\ell=0$,
where $\ell$ is a vector of the molecular mass of $m$ molecular
species.  Consider the aforementioned stoichiometrically
inconsistent example, where the corresponding stoichiometric matrices
are
\[
S:=R-F=\left[\begin{array}{cc}
0 & 1\\
0 & 0\\
1 & 0
\end{array}\right]-\left[\begin{array}{cc}
1 & 0\\
1 & 0\\
0 & 1
\end{array}\right]=\left[\begin{array}{rr}
\!\!-1 & 1\\
\!\!-1 & 0\\
1 & -1
\end{array}\right],
\]
with rows from top to bottom corresponding to molecular species
$A,B,C$. Let $a,b,c\in\mathbb{R}$ denote the molecular mass of $A,B,C$.
We require $a,b,c$ such that
\[
R^{T}\ell=\left[\begin{array}{ccc}
0 & 0 & 1\\
1 & 0 & 0
\end{array}\right]\left[\begin{array}{c}
a\\
b\\
c
\end{array}\right]=\left[\begin{array}{c}
c\\
a
\end{array}\right]=\left[\begin{array}{c}
a+b\\
c
\end{array}\right]=\left[\begin{array}{ccc}
1 & 1 & 0\\
0 & 0 & 1
\end{array}\right]\left[\begin{array}{c}
a\\
b\\
c
\end{array}\right]=F^{T}\ell.
\]
However, the only solution requires $a=c$ and $b=0$, i.e., a zero mass for
the molecule $B$, which is inconsistent with chemistry; therefore
the reactions $A+B\rightleftharpoons C$ and $C\rightleftharpoons A$
are stoichiometrically inconsistent.  In general, given $F$ and $R$,
one may check for stoichiometric consistency \citep{gevorgyan2008detection}
by solving the optimisation problem 
\[
\begin{aligned}\underset{\ell}{\max} & \;\left\Vert \ell\right\Vert _{0}\\
\text{s.t.} & \;S^{T}\ell=0,\\
 & \;0\leq\ell.
\end{aligned}
\]
Here, $\left\Vert \ell\right\Vert _{0}$ denotes the zero-norm or
equivalently the cardinality (number of non-zero entries) of $\ell$.
However, maximising the cardinality of a non-negative vector in the
left nullspace of $S$ is a problem that is challenging to solve exactly.
This problem has been represented as a mixed-integer linear optimisation
problem \citep{gevorgyan2008detection}, but since algorithms for
such problems have unpredictable computational complexity, we implemented
a novel and more efficient approach.

The cardinality of a non-negative vector is a quasiconcave
(or unimodal) function \citep{boyd_convex_2004}. The problem of maximising
this particular quasiconcave function, subject to a convex constraint,
may be approximated by a linear optimisation problem \citep{vlassis_fast_2014},
in our case the problem
\begin{equation}
\begin{aligned}\underset{z,\,\ell}{\max} & \;\mathbbm{1}^{T}z\\
\text{s.t.} & \;S^{T}\ell=0,\\
            & \;z \leq \ell,\\
            & \;0\le z \le \mathbbm{1}\alpha,\\
            & \;0\le \ell \le \mathbbm{1}\beta,
\end{aligned}
\label{eq:maxCardAApprox}
\end{equation}
where $z,\ell\in\mathbb{R}^{m}$ and $\mathbbm{1}$ denotes an all ones
vector. In this approximation, we maximise the sum over all dummy
variables $z_i$, $i=1,\dots,m$, but it is $\ell_i$ that represents the
stoichiometrically consistent molecular mass of the $i$th molecule.
The scalars $\alpha, \beta \in\mathbb{R} >0$ are proportional to the
smallest molecular mass considered non-zero and the largest molecular
mass allowed. An upper bound on the largest molecular mass avoids the
possibility of a poorly scaled optimal $\ell$. We used
$\alpha=10^{-4}$ and $\beta=10^{4}$ as all models tested were of
metabolism, so eight orders of magnitude between the least and most
massive metabolite is sufficient. As this approximation is based on
linear optimisation, it can be implemented numerically in a scalable
manner. We applied \eqref{eq:maxCardAApprox} to each reconstruction in
Supplementary Table 1 in order to identify stoichiometrically
inconsistent rows. That is, if $\ell^{\star}$ denotes the optimal
$\ell$ obtained from \eqref{eq:maxCardAApprox} then the $i$th row
is stoichiometrically inconsistent if
$\ell^{\star}_{i}<\alpha$. Stoichiometrically inconsistent rows and
the corresponding columns were omitted from further analyses. Where
molecular formulae were available, we confirmed that all retained
biochemical reactions were elementally balanced, as expected.  To
reiterate, in our numerical check of rank $[F,\,R]$, discussed below,
all rows correspond to metabolite species involved in
stoichiometrically consistent reactions, with the exception of
exchange reactions.

\subsubsection{Net flux consistency}

If one assumes that all molecules are at steady state, the corresponding
computational model should be \emph{net flux consistent,} meaning
that each net reaction of the network has a nonzero flux in at least
one feasible steady state net flux vector. Due to incomplete biochemical
knowledge, a reconstruction may contain \emph{net flux inconsistent
}reactions that do not admit a nonzero steady state net flux. For
example, consider the set of reactions
\begin{equation}
\emptyset\rightleftharpoons D\rightleftharpoons G\rightleftharpoons\emptyset,
\qquad D\rightleftharpoons H.\label{eq:fluxConsistentExample}
\end{equation}
In this set, the reaction $D\rightleftharpoons H$ is net flux inconsistent,
as any nonzero net flux is inconsistent with the assumption that the
concentration of $C$ should be time invariant.  Inclusion
of net flux inconsistent reactions, like $D\rightleftharpoons H$,
in a dynamic model would be perfectly reasonable, but we omit
such reactions because the focus of this paper is on modelling of steady
states.

Let $B\in\mathbb{R}^{m\times p}$ denote the stoichiometric matrix
for a set of $p$ exchange reactions. We say a matrix $S$ is net
flux consistent if there exist matrices $V\in\mathbb{R}^{n\times k}$
and $W\in\mathbb{R}^{p\times k}$ such that
\begin{eqnarray*}
SV & = & -BW,
\end{eqnarray*}
where each row of $V$ and each row of $W$ contains at least one
nonzero entry.  Consider the aforementioned net flux
inconsistent example, where the corresponding stoichiometric matrices
are 
\[
S=\left[\begin{array}{rr}
\!\!-1 & -1\\
1 & 0\\
0 & 1
\end{array}\right],\qquad
B=\left[\begin{array}{rr}
1 & 0\\
0 & -1\\
0 & 0
\end{array}\right].
\]
Let $p,q,r,s\in\mathbb{R}$ denote the net rate of the reactions,
from left to right in \eqref{eq:fluxConsistentExample}. We require
$p,q,r,s$ such that
\[
SV=\left[\begin{array}{rr}
\!\!-1 & -1\\
1 & 0\\
0 & 1
\end{array}\right]\left[\begin{array}{c}
p\\
q
\end{array}\right]=\left[\begin{array}{c}
\!\!-p-q\\
q\\
q
\end{array}\right]=\left[\begin{array}{c}
\!\!-r\\
s\\
0
\end{array}\right]=\left[\begin{array}{rr}
\!\!-1 & 0\\
0 & 1\\
0 & 0
\end{array}\right]\left[\begin{array}{c}
r\\
s
\end{array}\right]=-BW.
\]
However, the only solution requires $q=0$, i.e., a zero net flux
through the reaction $D\rightleftharpoons H$, corresponding to a zero
row of $V$; therefore this reaction is net flux inconsistent. Our
definition of net flux consistency is weaker than the assumption
that all reactions admit a nonzero net flux simultaneously, which
would be equivalent to requiring a single net flux vector with all
nonzero entries, i.e., $k=1$. It is also weaker than the assumption
of net flux consistency subject to bounds on the direction of reactions
\citep{vlassis_fast_2014}, which we do not impose here. Enforcing
net flux consistency requires omission of any net reaction that cannot
carry a non-zero net flux at a steady state.

Within \textsc{fastCORE}, a scalable algorithm for reconstruction of compact
and context-specific biochemical network models \citep{vlassis_fast_2014},
a key step employs linear optimisation as described above (\ref{eq:maxCardAApprox})
to identify the largest set of net flux consistent reactions in a
given model. We created a computational model from the stoichiometrically
consistent subset of each reconstruction in Supplementary Table 1.
We allowed all reactions to be reversible (lower and upper bounds
$-1000$ and $1000$), included exchange reactions in
each reconstruction, and then identified and omitted all net flux
inconsistent reactions ($v_{j}<\epsilon=10^{-4}$). We also omitted
the corresponding rows, where a molecular species is only involved
in flux inconsistent reactions. Therefore, in our check of rank($[F,\,R]$),
all rows correspond to metabolite species
involved in net flux consistent reactions. As Supplementary Table
1 illustrates, this is typically a subset of the stoichiometrically
consistent rows.

\subsubsection{Unique and non-trivial molecular species}

In a reconstruction, one may find a pair of rows in $S$ that are
identical up to scalar multiplication. As these extra rows typically
represent inadvertent duplication of an identical molecular species,
any such duplicate rows were omitted. Likewise, we omitted any row
with all zeros, e.g., corresponding to a metabolite that was only
involved in stoichiometrically inconsistent or net flux inconsistent
reactions. Hereafter, any biochemical network without zero rows or
rows identical up to scalar multiplication we refer to as being \emph{non-trivial}.

\subsubsection{\label{sub:Pervasive-flux-concentration-dua}Pervasive flux-concentration
duality}

\begin{sidewaysfigure}
\includegraphics[width=1\textwidth]{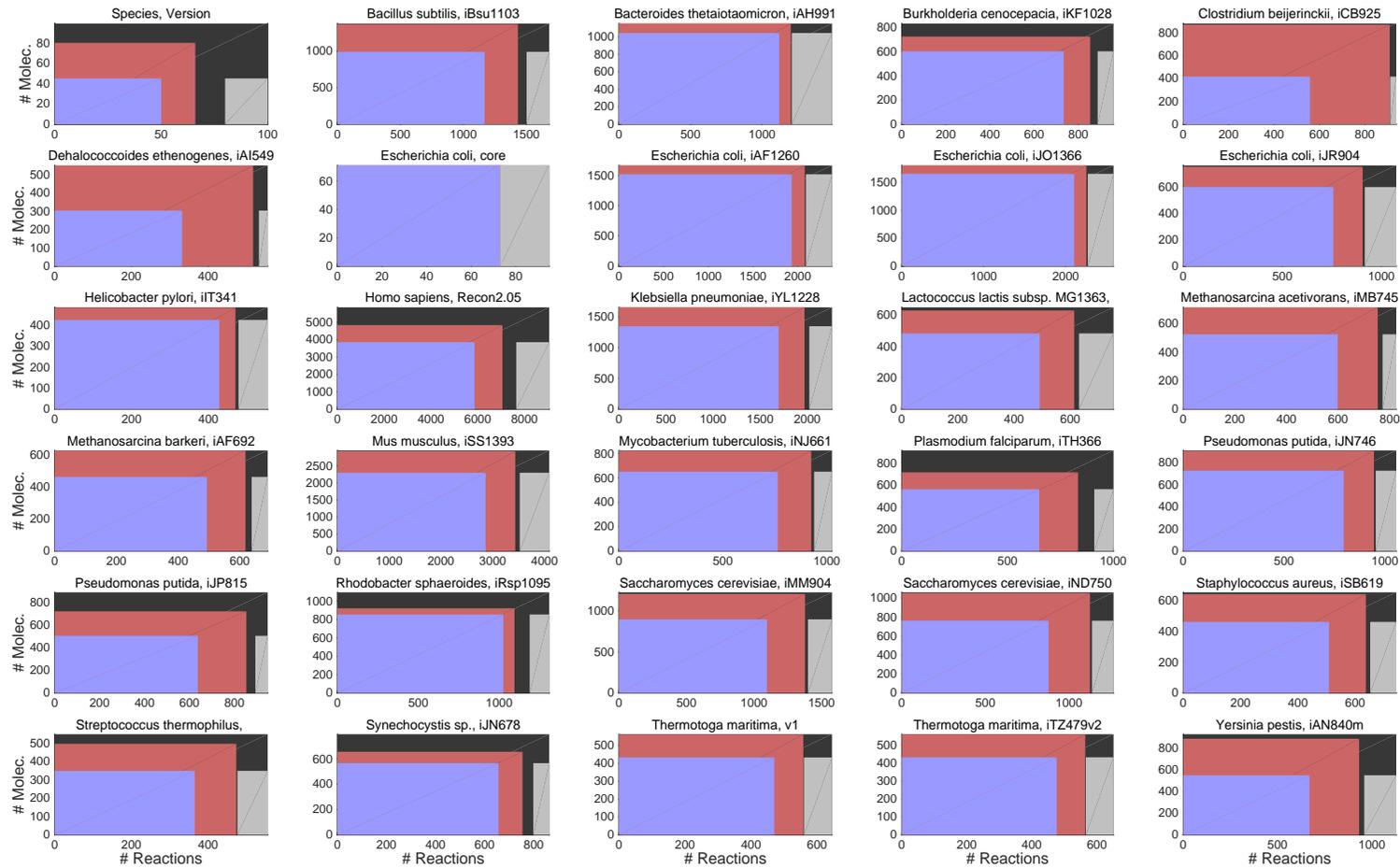}

\protect\caption{\label{fig:FRfillPlot}Usually, only a subset of a reconstruction
will satisfy the mathematical conditions imposed when a corresponding
computational model is generated. The original size of $[S,\,S_{e}]$
(outer black rectangle) varies across the 29 reconstructions tested.
Due to exchange reactions, only a subset of the columns of a reconstruction
correspond to stoichiometrically consistent rows (red rectangles).
If a molecular species is exclusively involved in exchange reactions,
the number of stoichiometrically consistent rows is less than the
number of rows of reconstruction. Due to reactions that do not admit
a nonzero steady state net flux, only a subset of mass balanced reactions
and a subset of exchange reactions are also flux consistent (blue
and grey rectangles, respectively). When $F$ and $R$ are derived
from a subset of a genome-scale biochemical network reconstruction,
assuming no zero rows of $[F,\,R]$ and no rows that are identical
up to scalar multiplication, stoichiometric and net flux consistency
is often but not always sufficient to ensure that $[F,\,R]$ has full
row rank.}
\end{sidewaysfigure}

We investigated the stoichiometric properties of a representative
subset of published metabolic network reconstructions. Specifically,
numerical experiments were performed on 29 published reconstructions
where a Systems Biology Markup Language \citep{Keating2006} compliant
Extensible Markup Language (.xml) file was available and at least
90\% of the molecular species corresponded to stoichiometrically consistent
rows.  Numerical linear algebra was used to compute
matrix rank (cf. Supplementary File 1, Section \ref{sub:Rank-Computation}).
The results are summarised in Figure \ref{fig:FRfillPlot} and provided
in detail in Supplementary File 2. All numerical experiments may be reproduced
with the {\sc Matlab} code distributed with the COBRA Toolbox at
\url{https://github.com/opencobra/cobratoolbox} (cf. Supplementary File 1, Section \ref{matlabres}).

The number of (possibly indistinct)
molecular species is, by definition, equivalent to the number of rows
of $S:=R-F$ derived directly from the reconstruction, without additional
assumptions. By forming $[F,\,R]$ directly from a reconstruction,
we found that $\textrm{rank}([F,\,R])$ is usually (21/29)
less than the number of rows of $S$, with some (8/29) exceptions,
e.g., the genome-scale reconstruction of the metabolic network of\emph{
Rhodobacter sphaeroides}, iRsp1095 \citep{21777427}. 

Most genome-scale reconstructions (26/29) were accompanied by chemical
formulae for the majority of reactions. If the number of stoichiometrically
consistent rows is less than the number of molecules exclusively involved
in reactions that are supposed to be elementally balanced, as determined
by a check for elemental balance, then at least one chemical formula
for a molecular species must be incorrectly specified. In only 3 of
the 26 reconstructions that supplied chemical formulae, this issue
was apparent (cf.~Supplementary file 1). Each reconstruction was converted
into a computational model where $F,\,R\in\mathbb{R}_{\geq0}^{m\times n}$
satisfy the following conditions:
\begin{enumerate}
\item All rows of $S:=R-F$ correspond to molecular species in stoichiometrically
consistent reactions, with the exception of exchange reactions.
\item No two rows in $[F,\,R]$ are identical up to scalar multiplication. 
\item All rows of $S$ correspond to molecular species in net flux consistent
reactions, assuming all reactions are reversible, including exchange
reactions.
\item No row of $[F,\,R]$ is all zeros.
\end{enumerate}

Of the 29 reconstructions subjected to the aforementioned conditions,
26 generated a model where $[F,\,R]$ had full row rank.
When $[F,\,R]$ was row rank-deficient, the rank was never
more than three less than the number of rows of $[F,\,R]$.
In each case, the rank-deficiency was a result of omitted biochemical
reactions that would otherwise have resulted in an $[F,\,R]$
with full row rank. A typical example of a genome-scale reconstruction
with row rank-deficient $[F,\,R]$ is highlighted in Section
\ref{sub:Combinatorial-dependence-in}.  In general,
should a row rank-deficient $[F,\,R]$ arise, there are
two options: (i) further manual reconstruction effort to correctly
specify reaction network stoichiometry, or (ii) omission of the dependent
molecular species from any derived kinetic model. 

Although conditions 2 and 4 are trivial and clearly necessary, neither
of conditions 1 or 3 (stoichiometric consistency or net flux
consistency) is necessary for $[F,\,R]$ to have full row rank. For
almost one third (8/29) of the reconstructions, one could form
$[F,\,R]$ without any further assumptions and yet $[F,\,R]$ had full
row rank. For instance, the genome-scale \emph{Methanosarcina
  acetivorans C2A} metabolic model (iMB745
\citep{Benedict:2012:J-Bacteriol:22139506}) has 715 molecular species
and without stoichiometric or net flux consistency being imposed,
$\textrm{rank}([F,\,R])=715$, even though this is 2 greater than the
number of stoichiometrically consistent rows of $S$.

When a stoichiometrically inconsistent row of $S$ is omitted from a metabolic model, the corresponding row of the biomass reaction is also omitted. This reduction in the number of constraints could lead to an increase in the maximum biomass synthesis rate. In contrast, removal of net flux inconsistent reactions might reduce the maximum biomass synthesis rate or render biomass synthesis infeasible. Flux balance analysis of each of the 29 genome-scale reconstructions before and after application of the aforementioned four conditions revealed that growth feasibility was not extinguished and tended to increase (data not shown). Further iterations of reconstruction and model validation would be required for each model derived in the manner described above prior to use in applications. In particular, one should check that each omitted reaction is balanced for each atomic element and conduct further literature research to resolve flux inconsistent reactions that contributed toward optimal biomass synthesis in models derived from reconstructions without the aforementioned quality control steps.

\section{Discussion}

Any net stoichiometric matrix $S\in\mathbb{R}^{m\times n}$ may be
derived by taking the difference between a pair of forward and reverse
stoichiometric matrices $F,R\in\mathbb{R}_{\geq0}^{m\times n}$, that
is $S:=R-F$. The horizontal concatenation $[F,\,R]\in\mathbb{R}^{m\times2n}$
is a key mathematical object that appears in the deterministic, elementary,
unidirectional reaction kinetic rate equation $v=\exp(\ln(k)+[F,\,R]^{T}\ln(c))$,
relating concentrations $c\in\mathbb{R}^{m}$ and rate coefficients
$k\in\mathbb{R}^{2n}$ to fluxes $v\in\mathbb{R}^{2n}$. We address
the question: When does there exist a one-to-one relationship between
concentrations and fluxes?

We have proven that, given rate coefficients, there is a one-to-one
relationship between concentrations and fluxes if and only if $[F,\,R]$
has full row rank. Furthermore, this dual relationship exists if and
only if there are no two disjoint sets of molecular species where
every corresponding unidirectional reaction involves at least one
molecular species from each of the disjoint sets. Flux-concentration
duality implies that one could discuss biochemistry either entirely
in terms of fluxes or entirely in terms of concentrations, as both
would be different perspectives on the same biochemical system. This
has clear implications when interpreting biochemical network function
from the perspective of either concentrations or fluxes.

Within a wide range of non-trivial biochemical network reconstructions,
including metabolism and signalling networks, we observe from numerical
experiments that together, stoichiometric and net flux consistency
of $S$ is often sufficient to ensure that $[F,\,R]$ has
full row rank. After application of these conditions we occasionally
observe that $[F,\,R]$ is row rank-deficient and this
is due to omission of reactions from the corresponding reconstruction.
Finding a numerical example where $[F,\,R]$ is row rank-deficient
does not reduce the biochemical significance of our observations if
the underlying network is not biochemically realistic. In each particular
case, it was clear that row rank-deficiency $[F,\,R]$
was due to the omission of known biochemical reactions that would
have given $[F,\,R]$ full row rank. It is easy to test
if $[F,\,R]$ has full row rank for a particular network,
but it is a rather abstract linear algebraic condition, so it is not
easy to see if it applies to biochemical networks in general. Therefore,
we sought a complementary characterisation of full-row-rank $[F,\,R]$
that was applicable in general and more easily interpretable from
a biochemical network perspective.

We have established biochemically interpretable combinatorial conditions
that are necessary and sufficient for $[F,\,R]$ to have
full row rank dependent only on the sparsity pattern of $F$ and $R$;
that is, independent of the actual values of their nonzero entries.
However, in practice these combinatorial conditions may be too strong,
because for any given biochemical network, the values of the nonzero
entries are fixed and the corresponding $[F,\,R]$ may
have full row rank, even if combinatorial independence of its rows
does not hold. Combinatorial independence of the rows of a given $[F,\,R]$
implies full row rank, but in general, the reverse implication does
not hold. In Section \ref{sub:exist_Flux-concentration-duality},
we applied numerical linear algebra to check the rank of $[F,\,R]$
derived from 29 reconstructions, each subject to certain conditions.
However, as the aforementioned $[F,\,R]$ all correspond
to networks of composite biochemical reactions, there exist columns
of $[F,\,R]$ with more than two nonzero entries. We do
not test for combinatorial independence of the rows of these $[F,\,R]$,
as this problem is NP-hard \citep{garey1979computers}. 

There are many interesting open problems, the solution of which would
be interesting extensions to this work. We know that all composite
reactions are defined from the composition of a set of elementary
reactions, and the latter give rise to an $[F,\,R]$ with
at most two nonzero entries in each column. Given an $[F,\,R]$
derived from a network of composite reactions, if one were to express
the network as a set of elementary reactions that properly reflects
the underlying biochemistry \citep{cook2007eka}, does the corresponding
$[F,\,R]$ also have full row rank?  One could ask the same
question starting from an elementary reaction network with an $[F,\,R]$
that has full row rank. Indeed, by Theorem \ref{thm:test}, testing
the combinatorial independence of the latter is solvable in polynomial
time. It is exciting that so many of the non-trivial, stoichiometrically
consistent and net flux consistent biochemical networks that we tested
do give rise to an $[F,\,R]$ of full row rank, despite
the fact that mathematically we know that these conditions are not
sufficient for $[F,\,R]$ to have full row rank. What are
the undiscovered, necessary, mathematical, yet biologically interpretable
conditions that ensure $[F,\,R]$ has full row rank, even
if its rows are not combinatorially independent? 

Putting this work into a broader context, one must always make a clear
distinction between a reconstruction and a model. In practice, the
latter is a numerical implementation that must satisfy certain mathematical
conditions that are usually not satisfied by every metabolite species
and every reaction in a given reconstruction. Indeed, depending on
one's combination of mathematical assumptions, one could derive many
different models from the same reconstruction. Testing for compliance
with mathematical conditions is a vital element of quality control
when converting a reconstruction into a correctly specified computational
model. Of note in this respect is the relatively low computational
complexity of the linear optimisation algorithms we use to solve the
problem of checking for stoichiometric and net flux consistency. 

Reconstruction mis-specification is often not due
to some error, especially for reconstructions that are ambitious in
scope.  Such reconstructions will inevitably contain knowledge gaps,
where the exact stoichiometry, chemical formula, etc, is unknown for
certain reactions. That is, reconstruction mis-specification is often
a reflection of incomplete biochemical knowledge. As any computational
model will only represent the subset of the metabolite species and
reactions that satisfy certain mathematical conditions, e.g., stoichiometric
consistency, one must take care to omit that part of a reconstruction
not satisfying certain conditions before generating model predictions
and absolutely before making any biological conclusions. Otherwise
grossly erroneous conclusions may be obtained. 

In applied mathematics, the development of an algorithm to find a
solution to a system of equations begins with certain assumptions
on the properties of the function(s) involved. In systems biochemistry,
deterministic modelling of molecular species concentrations gives rise
to systems of nonlinear equations, e.g., (\ref{eq:established_dcdt}),
the general mathematical properties of which are still being discovered.
Given rate coefficients, there is a paucity of scalable algorithms,
with guaranteed convergence properties, to solve large nonlinear biochemical
reaction equation systems for non-equilibrium, stationary concentrations.
Likewise for the problem of fitting optimal rate coefficients given
concentrations and a known reaction equation system. Observe that
(\ref{eq:established_dcdt}) contains the matrix $[F,\,R]$
twice and the matrix $[R,\,F]$ once. 

That $\textrm{rank}([F,\,R])=\textrm{rank}([R,\,F])=m$
is a pervasive property of biochemical networks from a diverse set
of organisms motivates the development of algorithms to exploit this
property and its consequences, e.g., \citep{artacho_accelerating_2015}.
This algorithmic development proceeds with two complementary approaches:
theory and numerical experiments.  Of particular importance
in this regard is that the set of models generated herein (with
$\textrm{rank}([R,\,F])=m)$
satisfy a common set of mathematical conditions, thereby reducing
the possibility for spurious numerical results, when numerically testing
hypothesised but unproven theorems concerning the properties of biochemical
networks in general. For instance it is known that a full row rank $[R,\,F]$ is a necessary but insufficient condition to preclude the existence of multiple positive steady states for certain chemical reaction networks \citep{muller_sign_2015}. Testing the rank of $[R,\,F]$ can be done efficiently, but it is still an open problem to design a tractable algorithm to test for the necessary and sufficient conditions to preclude the existence of multiple positive steady states for genome-scale biochemical networks \citep{muller_sign_2015}.
Numerical tests of a mathematical conjecture,
using biochemically realistic stoichiometric matrices, can be an efficient
way to find a counter-example or to provide support for the plausibility
of a conjecture. These tests help one decide where to invest
the mental effort required to attempt a proof of a conjecture. It
is important therefore that such numerical tests be conducted with
(a) a wide selection of stoichiometric matrices, in case a conjecture
holds only for certain network topologies, and (b) a set of stoichiometric
matrices that each satisfy a specified set of biochemically motivated
mathematical conditions, in case a conjecture holds only for stoichiometric
matrices corresponding to realistic biochemical networks.

\section{Conclusions}

Mathematical and computational modelling of biochemical networks is
often done in terms of either the concentrations of molecular species
or the fluxes of biochemical reactions. Mathematical modelling from
either perspective is equivalent when concentrations and unidirectional
fluxes are dual variables.  Assuming elementary kinetic
rate laws for each reaction, we show that this duality holds if and
only if the matrix $[F,\,R]\in\mathbb{R}_{\geq0}^{m\times2n}$
has full row rank, where $[F,\,R]$ is formed by concatenation
of the stoichiometric matrices $F,R\in\mathbb{R}_{\geq0}^{m\times n}$
for the $m$ reactants consumed in $n$ forward and reverse reaction
directions, respectively. Numerical experiments with computational
models derived from many genome-scale biochemical networks indicate
that flux-concentration duality is a pervasive property of biochemical
networks. For an arbitrary biochemical network, we provide a combinatorial
characterisation that is sufficient to ensure flux-concentration duality.
That is, for every two disjoint sets of molecular species, if there
is at least one reaction complex that involves species from only one
of the two sets, then duality holds. Our stoichiometric characterisation
of the conditions for duality between concentrations and unidirectional
fluxes has fundamental implications for mathematical and computational
modelling of biochemical networks. When flux-concentration duality
holds, interpretation of biochemical network function from the perspective
of unidirectional fluxes is equivalent to interpretation from the
perspective of molecular species concentrations.

\section*{Acknowledgements}

We would like to thank Michael Tsatsomeros and Francisco J. Aragon
Artacho for valuable comments. This work was funded by the Interagency
Modeling and Analysis Group, Multiscale Modeling Consortium U01
awards from the National Institute of General Medical Sciences [award
GM102098] and U.S. Department of Energy, Office of Science, Biological
and Environmental Research Program [award DE-SC0010429].  The content
is solely the responsibility of the authors and does not necessarily
represent the official views of the funding agencies.

\section*{References}
\normalsize

\bibliographystyle{elsarticle-harv-nourl}


\pagebreak{}

\normalsize

\section{Supplementary Material}
\setcounter{page}{1}

\subsection{\label{sec:Numerical-Linear-algebra}Numerical Linear algebra}

\subsubsection{\label{sub:Rank-Computation}Rank Computation}

To compute the rank of a matrix, we employed the sparse LU factorisation
package LUSOL \citep{gill_maintaining_1987,gill_snopt:_2005}. Given
a sparse matrix $A\in R^{m\times n}$, LUSOL computes factorisations
of the form 
\begin{equation}
P_{1}AP_{2}=LDU,\label{eq:LDU}
\end{equation}
where $P_{1}$ and $P_{2}$ are permutations, $L$ is lower trapezoidal
with unit diagonals, $D$ is diagonal and nonsingular, $U$ is upper
trapezoidal with unit diagonals, and the rank of each factor $L$,
$D$, $U$ is $r\le\min(m,n)$. This is LUSOL's estimate of rank($A$).

The permutations are chosen to keep $L$ and $U$ sparse, subject
to bounds on the off-diagonal elements of $L$ and $U$. Threshold
Partial Pivoting \citep{gill_maintaining_1987} requires $|L_{ij}|\le\tau$
for some tolerance $\tau\in(1,10]$, where $\tau=2$ is a reasonable
choice. Threshold Rook Pivoting \citep{gill_maintaining_1987} also
requires $|U_{ij}|\le\tau$ and is likely to be more reliable on general
sparse $A$. We have observed that net stoichiometric matrices $S=R-F$
have a sharply defined rank that can be estimated reliably by Threshold
Partial Pivoting on either $S$ or $S^{T}$ and this is cheaper than
applying Threshold Rook Pivoting. The same is true for estimating
the rank of $[F,\,R]$.

\subsubsection{\label{sub:Identification-of-dependencies}Identification of dependencies}

To investigate the rationale for row rank-deficiency of $A:=[F,\,R]$
derived directly from a reconstruction, we used Threshold Partial
Pivoting to obtain the factorisation (\ref{eq:LDU}). The row permutation
$P_{1}$ partitions the rows as 
\[
P_{1}A=\begin{bmatrix}B\\
C
\end{bmatrix},
\]
where $B\in R^{r\times n}$, $C\in R^{(m-r)\times n}$, and rank($B$)
= $r$. By definition of rank, the over-determined linear system $B^{T}W=C^{T}$
is consistent (has a solution $W$). The nonzero entries of each column
of $W$ reveal the dependencies between rows of $B$ and $C$. We
obtained $W$ by solving $\min\|B^{T}W-C^{T}\|$ using sparse QR factorisation
(\texttt{W = B'$\backslash$C';} in {\sc Matlab}).
The column permutation $P_{2}$ further partitions $A$ as 
\[
P_{1}AP_{2}=\begin{bmatrix}B_{1} & B_{2}\\
C_{1} & C_{2}
\end{bmatrix},
\]
where $B_{1}$ is $r\times r$. We could obtain $W$ more efficiently
by solving $B_{1}^{T}W=C_{1}^{T}$, where $B_{1}=L_{1}U_{1}$ is already
factorised in terms of the first $r$ rows and columns of $L$ and
$U$. The nonzero entries of $W$ reveal the dependencies between
rows of $B_{1}$ and $C_{1}$. As each row of $B_{1}$ and $C_{1}$
corresponds to a different molecular species, one can use $W$ to
investigate the biochemical rationale for dependency among rows of
$[F,\,R]$ if dependency is observed.
\pagebreak{}

\subsection{\label{sub:Combinatorial-dependence-in}Combinatorial
dependence in exceptional models derived from genome-scale reconstructions}

Combinatorial dependence among the rows of $[F,\,R]$ implies that $[F,\,R]$ is row rank deficient. However the reverse implication is not necessarily true. This is because linear dependence depends on the nonzero numerical values of elements in $[F,\,R]$ whereas combinatorial dependence depends only on the sparsity pattern of $[F,\,R]$ and not the numerical values in $[F,\,R]$. Nevertheless, it is of interest to check if an $[F,\,R]$ that is row rank deficient also contains combinatorially dependent rows.

Only 3 of the 29 reconstructions subjected to the four conditions
in Section \ref{sub:Pervasive-flux-concentration-dua} resulted in
a row rank-deficient $[F,\,R]$, with rank at most 3 lower
than the number of rows. If rank($[F,\,R]$) is less
than the number of rows of $[F,\,R]$ then numerical linear
algebra (cf. Section \ref{sub:Identification-of-dependencies}) can
be used to test for dependency between rows to identify possible reasons
for the row rank deficiency. The three models with rank-deficient
$[F,\,R]$ were from compartmentalised genome-scale models.
In each of the three models, one could find at least one dependency
between a dependent molecular species (one dependent row of $[F,\,R]$)
and a set of independent molecular species within the same sub-cellular
compartment (a set of linearly independent rows of $[F,\,R]$). 

Each dependency in $[F,\,R]$ was due to the existence
of two disjoint sets of molecular species, one having a cofactor moiety
in common and one having a non-cofactor moiety in common, with one
cofactor and one non-cofactor always present in reactant complexes
within that sub-cellular compartment. That is, neither the cofactor
moiety nor the non-cofactor moiety was synthesised or degraded within
that sub-cellular compartment. This reflected one or more reactions
that were missing from the original reconstruction that would either
synthesise or degrade the moiety within that sub-cellular compartment.
Such reactions would give $[F,\,R]$ full row rank, as
inevitably there would be at least one reaction where the cofactor
and non-cofactor moiety would not both be represented within a reaction
complex at the same time. Table~\ref{tab:exceptionExample} illustrates
a specific example of such a dependency within a model derived from
a \emph{Saccharomyces cerevisiae} reconstruction (iMM904).
Alternatively, the dependency reflected the omission of a reaction
to transport either the cofactor or non-cofactor moiety into that
sub-cellular compartment. Such a reaction would typically not simultaneously
involve both the cofactor moiety and the non-cofactor moiety rendering
$[F,\,R]$ of full row rank. 

\begin{table}
\protect\caption{\label{tab:exceptionExample}An example of a set of endoplasmic
reticulum reactions within a genome-scale \emph{Saccharomyces cerevisiae}
reconstruction (iMM904), that results in a rank deficient $[F,\,R]$
even when all molecular species are stoichiometrically consistent
and all reactions are net flux consistent ( given exchange reactions,
assuming each reaction is reversible and omitting nontrivial rows).
Besides the 6 reactions illustrated, Phytosphingosine, Sphinganine,
Tetracosanoyl-CoA, Hexacosanoyl-CoA and Phosphate are not involved
in any other reactions within the endoplasmic reticulum. Phytosphingosine
and Sphinganine form one set reactions, while Tetracosanoyl-CoA, Hexacosanoyl-CoA
and Phosphate form another set. These sets are disjoint as they do
not share a molecular species in common. Observe that the support
of both disjoint sets is identical, i.e., all reactions contain at
least one member of both sets in the same reaction complex. This renders
the corresponding rows of $[F,\,R]$ row rank deficient.
In fact, the rank of these 5 rows is 4, hence leading to row rank
deficiency of $[F,\,R]$. A more comprehensive model would
have the 3-ketodihydrosphingosine reductase reaction (Phytosphingosine
+ NADPH $\rightleftharpoons$ 3-Dehydrosphinganine + NADP) that would
result a full row rank $[F,\,R]$.}

\medskip

\begin{small}
\begin{tabular}{ll}
   R1 & alkaline ceramidase (ceramide-1)
\\ R2 & alkaline ceramidase (ceramide-1)
\\ R3 & alkaline ceramidase (ceramide-2)
\\ R4 & alkaline ceramidase (ceramide-2)
\\ R5 & sphingoid base-phosphate phosphatase (sphinganine 1-phosphatase)
\\ R6 & sphingoid base-phosphate phosphatase (phytosphingosine 1-phosphate)
\\ A1 & CERASE124er
\\ A2 & CERASE126er
\\ A3 & CERASE224er
\\ A4 & CERASE226er
\\ A5 & SBPP1er
\\ A6 & SBPP1er
\end{tabular}

\begin{center}
\begin{tabular}{|l|l|c|c|c|c|c|c|}
\hline 
 & Reaction Name & R1 & R2 & R3 & R4 & R5 & R6
\tabularnewline
\hline
 Metabolite species name & Abbreviation & A1 & A2 & A3 & A4 & A5 & A6
\tabularnewline
\hline 
\hline 
Ceramide-1 (Sphinganine:n-C24:0) & cer1\_24{[}r{]} & -1 & 0 & 0 & 0 & 0 & 0\tabularnewline
\hline 
Ceramide-1 (Sphinganine:n-C26:0) & cer1\_26{[}r{]} & 0 & -1 & 0 & 0 & 0 & 0\tabularnewline
\hline 
Ceramide-2 (Phytosphingosine:n-C24:0) & cer2\_24{[}r{]} & 0 & 0 & -1 & 0 & 0 & 0\tabularnewline
\hline 
Ceramide-2 (Phytosphingosine:n-C26:0) & cer2\_26{[}r{]} & 0 & 0 & 0 & -1 & 0 & 0\tabularnewline
\hline 
Coenzyme A & coa{[}r{]} & -1 & -1 & -1 & -1 & 0 & 0\tabularnewline
\hline 
Proton & h{[}r{]} & -1 & -1 & -1 & -1 & 0 & 0\tabularnewline
\hline 
Water & h20{[}r{]} & 0 & 0 & 0 & 0 & -1 & -1\tabularnewline
\hline 
Sphinganine 1-phosphate & sph1p{[}r{]} & 0 & 0 & 0 & 0 & -1 & 0\tabularnewline
\hline 
Phytosphingosine 1-phosphate & psph1p{[}r{]} & 0 & 0 & 0 & 0 & 0 & -1\tabularnewline
\hline 
Phytosphingosine & psphings{[}r{]} & 0 & 0 & \textbf{1} & \textbf{1} & 0 & \textbf{1}\tabularnewline
\hline 
Sphinganine & sphgn{[}r{]} & \textbf{1} & \textbf{1} & 0 & 0 & \textbf{1} & 0\tabularnewline
\hline 
Tetracosanoyl-CoA & ttccoa{[}r{]} & \textbf{1} & 0 & \textbf{1} & 0 & 0 & 0\tabularnewline
\hline 
Hexacosanoyl-CoA (n-C26:0CoA) & hexccoa{[}r{]} & 0 & \textbf{1} & 0 & \textbf{1} & 0 & 0\tabularnewline
\hline 
Phosphate & pi{[}r{]} & 0 & 0 & 0 & 0 & \textbf{1} & \textbf{1}\tabularnewline
\hline 
\end{tabular}
\end{center}
\end{small}
\end{table}

\clearpage

\subsection{\label{matlabres}Reproduction of numerical results}

All of the reconstructions and code required to reproduce
the numerical results referred to in this paper are publicly available
within the COBRA toolbox \citep{cobraV2supp} via
\url{https://github.com/opencobra/cobratoolbox}.
The steps are as follows:
\begin{enumerate}

\item Install {\sc Matlab} version 8.4.0.150421 (R2014b) or above.
Earlier versions of {\sc Matlab} may also suffice, but have not been tested
for this purpose.

\item Install the latest version of The COBRA toolbox (more
recent than December 1, 2015). From a unix command line, enter the
command: 
\\
git clone https://github.com/opencobra/cobratoolbox.git

\item Optionally install a 64-bit Unix implementation of
LUSOL
\\ 
\url{http://stanford.edu/group/SOL/software/lusol/}.
\\
From a Unix command line, enter the command 
\\
git clone \url{https://github.com/nwh/lusol.git}.
\\
Installation is optional as otherwise the sparse LU factorization
provided in {\sc Matlab} is employed.

\item The folder \url{cobratoolbox/testing/testModels/modelCollectionFR}
contains each of the reconstructions in COBRA Toolbox format (one
{\sc Matlab} .mat file for each reconstruction). Each .mat file was derived
from the original SBML file that was published with the respective
papers or provided as published updates to the original SBML file,
as detailed within the function 
\\
\url{cobratoolbox/testing/testModels/modelCitations.m}.

\item All numerical results can be reproduced by calling
the {\sc Matlab} function
\\
\url{cobratoolbox/papers/Fleming/FR_2015/checkRankFRdriver.m}.
\\
This driver file passes each reconstruction to \url{checkRankFR.m}, which
generates the corresponding model as detailed in Section 
2.4.4
 and uses numerical linear algebra to check rank($[F,\,R]$),
as described in Section \ref{sub:Rank-Computation}.
\end{enumerate}

\section*{Supplementary references}
\footnotesize

\bibliographystyle{elsarticle-harv-nourl}


\end{document}